\documentclass[10pt]{amsart}

\usepackage{amssymb,amsthm,amsmath}
\usepackage[numbers,sort&compress]{natbib}
\usepackage{color}
\usepackage{graphicx}
\usepackage{tikz}

\title[ ]{A  lower bound on  the Lyapunov exponent for  the generalized
  Harper's model}
\author{ Svetlana Jitomirskaya}
\address[ Svetlana Jitomirskaya]{ Department of Mathematics, University of California, Irvine, California 92697-3875, USA}
\email{szhitomi@math.uci.edu}

\author{Wencai Liu}
\address[Wencai Liu]{Department of Mathematics, University of California, Irvine, California 92697-3875, USA} \email{liuwencai1226@gmail.com}

\newcommand{\R}{\mathbb{R}}

\theoremstyle{plain}
\newtheorem{theorem}{Theorem}[section]

\newtheorem{lemma}[theorem]{Lemma}

\theoremstyle{definition}

\newtheorem{remark}[theorem]{Remark}

\begin{document}


\begin{abstract}
We obtain a
  lower bound for the  Lyapunov exponent  of  a family of discrete Schr\"{o}dinger operators $(Hu)_n=u_{n+1}+u_{n-1}+2a_1\cos2\pi(\theta+n\alpha)u_n+2a_2\cos4\pi(\theta+n\alpha)u_n$,
that incorporates both $a_1$ and $a_2,$ thus going beyond the Herman's bound. 

\end{abstract}
\maketitle
\centerline{ \small \it {Dedicated to David Ruelle and Yasha Sinai on
    the occasion of their 80th birthdays}}
 \section{Introduction}

While Lyapunov exponents of analytic one-frequency Schr\"odinger cocycles are positive when the
coupling constants are large (uniformly in frequency and in
energy)\cite{ss,bourgain2005green}, in general, there have been no explicit/quantitative
lower estimates. The only such known explicit estimates are celebrated
Herman's subharmonicity bounds: for $v(x)=\sum_{j=-M}^Ma_j e^{ 2\pi
  jx}, a_j=\bar{a}_{-j},$
there is a uniform bound: $L(E)\geq\ln |a_M|$. This is of course not
effective if $|a_M|$ is small. In this paper we present the first
explicit bound that takes into account other coefficients.

Namely, we consider  the discrete Schr\"{o}dinger operator  of the following form:
\begin{equation}\label{Def.Schrodingeroperator2}
     (Hu)_n=u_{n+1}+u_{n-1}+2a_1\cos2\pi(x+n\alpha)u_n+2a_2\cos4\pi(x+n\alpha)u_n,
\end{equation}
where $a_1,a_2\in\mathbb{R}$. 
It is known in physics literature as generalized Harper's model
(e.g. \cite{23,38}), of special  interest because of its connection to
the three dimensional quantum Hall effect\cite{23,35}.
\par
 Analyzing the solution of the equation  $Hu=Eu$ from the dynamical
 systems point of view  leads to the consideration of
   Schr\"{o}dinger   cocycle $(\alpha,A)$, where
 \begin{equation}\label{Def.cocycle}
   A(x)=\left(
          \begin{array}{cc}
            E- 2a_1\cos2\pi x-2a_2\cos4\pi x &  -1 \\
            1 & 0 \\
          \end{array}
        \right).
 \end{equation}
    The  Lyapunov exponent(LE) of $(\alpha,A)$ is given  by
 \begin{equation*}
   L^E=\lim_{n\rightarrow\infty} \frac{1}{n}\int_{\mathbb{R} / \mathbb{Z}} \ln \|A_n(x)\|dx,
 \end{equation*}
where
\begin{equation*}
     A_n(x) = A(x+(n-1)\alpha)A(x+(n-2)\alpha)\cdots A(x) \text{ for } n\geq 0
\end{equation*}
and $A_{-n}(x)=A_n(x-n\alpha)^{-1}$.  We call $A_n(x)$   the $n$ step transfer matrix.
Notice that $A$ depends on $E$; sometimes we will drop this dependence
from the notations, for simplicity.
\par

\par
By Herman's arguments \cite{herman1983methode},
we have  that $L^E\geq \ln |a_2|$ for all $E\in \R$. This yields  that
$L^E>0$ if $|a_2|>1$. Also, $L^E\geq \ln
 |a_1|$ if $a_2=0$. By the continuity of Lyapunov exponent in the cocycle (Theorem
 \ref{ContinuousLE}), we still have     $L^E> 0$ if $|a_1|>1$ and
 $a_2$ is small. However, Herman's method  does not work to obtain a
 lower bound no matter how small  a nonzero $a_2$ is . In
 \cite{ss,bourgain2005green} a lower bound is obtained abstractly,
 for any analytic $v$ in terms of its analytic extension, but in both
 cases in ways that so far have not led to concrete analytic bounds.   In this paper we obtain a
 quantitative lower bound on the Lyapunov exponent,
in terms of $\ln |a_1|$ and $|\frac{a_2}{a_1}|^{\frac{1}{2}}$.
 \par
 \begin{theorem}\label{Maintheorem}
 Suppose $|a_1|>1 $ and $|a_2|<\frac{1}{100}|a_1|$.
Then  for all $E\in \R$,
\begin{equation}\label{Lowerbound}
    L^E\geq  \ln|a_1|-10|\frac{a_2}{a_1} |^{\frac{1}{2}}.
\end{equation}

 \end{theorem}

 Our result is obtained by combining the basics of Avila's global
 theory   \cite{avila2009global} with Herman's subharmonicity  method
 \cite{herman1983methode} as developed by Bourgain \cite{bourgain2005green}, coupled with some
 elementary geometric considerations.

{\bf Remark.}
After our result was obtained, we learned that  Marx, Shou, and Wellens     obtained  a general lower bound
for the Lyapunov exponent improving Herman's results, also using the
global theory \cite{marx}.  However it is  difficult  to extract
analytic quantitative results for a concrete potential from their
estimates, although numerical results are possible, leading to
estimates better than ours for certain parameters.
 \section{Avila's Global Theory}
 Denote by $  C^{\omega}(\mathbb{T},\text{SL}(2,C))$ the class of $1-$periodic functions on $\mathbb{C}$,
with analytic extension to some strip $|\Im z|<\delta$, attaining values in  $\text{SL}(2,C)$. Notice that matrix $A$  given by (\ref{Def.cocycle}) is in
$  C^{\omega}(\mathbb{T},\text{SL}(2,C))$.
The following theorem shows the continuity of LE in the analytic category.
\begin{theorem}(\cite{jitomirskaya2005localization,bj})\label{ContinuousLE}
The Lyapunov exponent $ L(\alpha,D)$  is  a continuous function in
$\mathbb{R}\times C^{\omega}(\mathbb{T} ,\text{SL}(2,C))$ at every
$(\alpha,A)$ with $\alpha \in \mathbb{R}\backslash \mathbb{Q}$.
\end{theorem}
The global theory   is based on the complexification of the cocycle. More precisely, given any analytic  cocycle $(\alpha,D)$,
we consider its analytic  extension, $(\alpha,D(x+i\epsilon))$  with  $|\epsilon|<\delta$ for some appropriate $\delta>0$.
$L(\alpha,D_\epsilon)$ is referred to the Lyapunov exponent of the  complexified cocycle $(\alpha,D(x+i\epsilon))$. Let $L(\alpha,D)=L(\alpha,D_0)$.
\par
The   Lyapunov exponent $L(\alpha,D_\epsilon)$ is known  to be a
convex function of $\epsilon$ (e.g. Proposition A.1, \cite{jitomirskaya2012analytic}). Thus we can introduce the acceleration of
$(\alpha,D)$,
\begin{equation*}
      \omega(\alpha,D;\epsilon)=\frac{1}{2\pi}\lim_{h\rightarrow0+}\frac{L(\alpha,D_{\epsilon+h})-L(\alpha,D_\epsilon)}{h}.
\end{equation*}
It follows from convexity and continuity of the Lyapunov exponent that the acceleration is an upper semi-continuous function  with respect to  parameter $\epsilon$
\cite{jitomirskaya2012analytic}.

The acceleration has an important property, which is stated below.
\begin{theorem}(Quantization of acceleration\cite{avila2009global,avila2013complex})\label{QAT}
Suppose $ D\in C^{\omega}(\mathbb{T},\text{SL}(2,C))$ and $\alpha\in \mathbb{R}\backslash \mathbb{Q}$, then we have
$ \omega(\alpha,D;\epsilon)\in  \mathbb{Z}$.
\end{theorem}
 Applying Theorem \ref{QAT} to Schr\"{o}dinger cocycle $(\alpha,A)$, where $A$ is given by (\ref{Def.cocycle}), we can obtain the following lemma. Below, we always assume
  $A$ is given by (\ref{Def.cocycle}),   $L^E(\epsilon)=L(\alpha,A(\cdot+i\epsilon))$ and $L^E=L^E(0)$.
  By the symmetry ( $L^E(\epsilon)=L^E(-\epsilon)$ ) and convexity of
  $L^E(\epsilon)  $,   we know  that $ \omega(\alpha,D; 0)\geq 0$.
\begin{lemma}
For $\alpha\in \mathbb{R}\backslash\mathbb{Q}$,
there exists an $\epsilon'>0$ such that $L^E(\epsilon)=4\pi \epsilon+\ln|a_2|$ for $\epsilon>\epsilon'$.
\end{lemma}
\begin{proof}
First, uniformly for $x\in \mathbb{T}={\mathbb{R} / \mathbb{Z}}$,  one has
 \begin{equation*}
        A(x+i\epsilon)= a_2e^{4\pi \epsilon} e^{-4\pi ix}(A_{\infty}+o(1)),
 \end{equation*}
 as $\epsilon\rightarrow +\infty$,
 where
 \begin{equation*}
    A_{\infty}=\left(
                         \begin{array}{cc}
                           -1 &0 \\
                           0 &0 \\
                         \end{array}
                       \right).
 \end{equation*}
 By the continuity of the LE,
 \begin{equation*}
    L^E(\epsilon)= 4\pi \epsilon+\ln|a_2|,
 \end{equation*}
  as $\epsilon\rightarrow +\infty$.
  \par
  By the quantization of acceleration, we obtain
   \begin{equation*}
   L^E(\epsilon)= 4\pi \epsilon+\ln|a_2|, \text{ for all }   \epsilon>0  \text{ sufficiently large}.
  \end{equation*}

\end{proof}
In addition to the convexity and continuity of LE with respect to $\epsilon$,  we have that
the graph $(\epsilon,L^E(\epsilon))$   must be one of the four
pictures below (we assume here $L^E>0 $).
\par

\begin{tikzpicture}

\draw [->](-0.5,0)--(4.5,0);
\draw [->](0,0)--(0,4.5);
\draw (0,0.45)--(1.5,1.8);
\draw [->](1.5,1.8)--(4,4);
\node [right] at (0,4){$L^E $};
\node [above] at (2.3,2.3){$\ell_2$};
\node [below] at (0,-0.3){Fig.1};
\node [below] at (4,0){$\epsilon$};
\end{tikzpicture}
\;\;\;\; \;\;\;\; \;\;\;\; \;\;\;\; \;\;\;\; \;\;\;\;\;\;\;\; \;\;\;\;
\begin{tikzpicture}
\draw [->](-0.5,0)--(4.5,0);
\draw [->](0,0)--(0,4.5);
\draw (0,1)--(1.5,1.8);
\draw [->](1.5,1.8)--(4,4);
\node [right] at (0,4){$L^E $};
\node [above] at (0.7,0.9){$\ell_1$};
\node [above] at (2.3,2.0){$\ell_2$};
\node [below] at (0,-0.3){Fig.2};
\node [below] at (4,0){$\epsilon$};

\end{tikzpicture}
\par
\begin{tikzpicture}
\draw [->](-0.5,0)--(4.5,0);
\draw [->](0,0)--(0,4.5);
\draw (0,1)--(0.5,1);
\draw (0.5,1)--(1.5,1.8);
\draw [->](1.5,1.8)--(3.5,4.5);
\node [right] at (0,4){$L^E $};
\node [above] at (1.0,0.9){$\ell_1^{\prime}$};
\node [above] at (2.5,2.5){$\ell_2$};
\node [below] at (0,-0.3){Fig.3};
\node [below] at (4,0){$\epsilon$};
\node [above] at (0.3,0.9){$\ell_3$};
\end{tikzpicture}
\;\;\;\; \;\;\;\; \;\;\;\; \;\;\;\; \;\;\;\; \;\;\;\;\;\;\;\; \;\;\;\;
\begin{tikzpicture}
\draw [->](-0.5,0)--(4.5,0);
\draw [->](0,0)--(0,4.5);
\draw (0,1)--(1,1);
\draw [->](1,1)--(3.5,4.5);
\node [right] at (0,4){$L^E $};
\node [above] at (2.5,2.5){$\ell_2$};
\node [below] at (0,-0.3){Fig.4};
\node [below] at (4,0){$\epsilon$};
\node [above] at (0.5,0.9){$\ell_3$};
\end{tikzpicture}
 $\ell_2:L^E(\epsilon)= 4\pi \epsilon+\ln|a_2|$,  $\ell_1:L^E(\epsilon)= 2\pi \epsilon+L^E$,  $\ell_1^{\prime}:L^E(\epsilon)= 2\pi \epsilon+C$, $\ell_3:L^E(\epsilon)= L^E$.
 \par
 We only present the graph for $L^E>0$, for $L^E=0$, it is similar.
  \begin{remark}
 Note that the graph $(\epsilon,L^E(\epsilon))$   depends on $E$. The graph can determine the  spectrum $\sigma(H)$.
 \begin{itemize}
 \item
 $\{E\in\R: L^E=0 \}\subset \sigma(H)$.
 \item
 If  $L^E>0$, then  $E\in \sigma(H)$ if and only if the graph of  $L^E(\epsilon)$ is either Fig.1 or Fig.2\cite{avila2009global}.
\end{itemize}
 \end{remark}
\section{The proof of  Theorem \ref{Maintheorem}}

Before giving  the proof of Theorem \ref{Maintheorem}, two lemmas are necessary.
\begin{lemma}\label{EstimateA_j}
Suppose $B_j=\left(
               \begin{array}{cc}
                 v_j & -1 \\
                 1 & 0 \\
               \end{array}
             \right),
$ and $|v_j|>2$, then we have
\begin{equation*}
  || \prod_{j=1} ^nB_j||\geq  \prod_{j=1} ^n(|v_j|-1 ) .
\end{equation*}

\end{lemma}
\begin{proof}
  Lemma \ref{EstimateA_j} can be proved by induction. For this purpose,
  let
  \begin{equation*}
       \prod_{j=1} ^kB_j=\left(
                          \begin{array}{cc}
                            b_{11}^k &  b_{12}^k \\
                             b_{21}^k &  b_{22}^k \\
                          \end{array}
                        \right).
  \end{equation*}
Actually,  we can prove  the following two inequalities,
 \begin{equation}\label{induction1}
  |b_{11}^k|\geq  \prod_{j=1} ^k(|v_j|-1 )
 \end{equation}
  and
 \begin{equation}\label{induction2}
  |b_{21}^k|\leq  |b_{11}^k|.
 \end{equation}
   Clearly, (\ref{induction1}) and (\ref{induction2}) hold for $k=1$. 
   Suppose  (\ref{induction1}) and (\ref{induction2}) hold for $k$,
  then  by    induction,
  \begin{eqnarray*}
    |b_{11}^{k+1}| &=& |   v_{k+1}  b_{11}^{k}-b_{21}^k| \\
     &\geq&    |  v_{k+1}||b_{11}^{k}|-|b_{11}^k|\\
      &\geq&   \prod_{j=1}^{k+1}  (|v_j|-1),
  \end{eqnarray*}
  and
  \begin{equation*}
    |b_{21}^{k+1}|=|b_{11}^k|\leq  |b_{11}^{k+1}|.
  \end{equation*}
  This implies (\ref{induction1}) and (\ref{induction2}) hold for $k+1$.
  \par
\end{proof}
\begin{lemma}\label{LeMinestimate}
Suppose $e^{4\pi \epsilon_0}\geq 10$ and $a_1>0$. Then
for any $E\in \mathbb{R}$,we have 
 \begin{equation*}
   \sup_{\delta\in \{\epsilon_0,2\epsilon_0\}} \inf_{x\in\mathbb{R}}|E  -a_1e^{-2\pi \delta}e^{2\pi ix} -a_1e^{2\pi \delta}e^{-2\pi ix} |\geq \frac{19}{60}a_1 e^{4\pi \epsilon_0} .
 \end{equation*}
\end{lemma}
\begin{proof}
Let $ a_\delta=  a_1e^{2\pi \delta}+a_1e^{-2\pi \delta}$ and $ b_\delta= a_1e^{2\pi \delta}-a_1e^{-2\pi \delta}$ with $\delta\in \{\epsilon_0,2\epsilon_0\}$.
Then  it is easy to check that
\begin{equation*}
 \inf_{x\in\mathbb{R}}|E  -a_1e^{-2\pi \delta}e^{2\pi ix} -a_1e^{2\pi \delta}e^{-2\pi ix} |
\end{equation*}
 is equal to the
distance $\mbox{dist}(E,S_\delta)$ between point $(E,0)$ and ellipse
$S_\delta$ given by
\begin{equation*}
\frac{x^2}{a_\delta^2}+  \frac{y^2}{b_\delta^2}=1.
\end{equation*}
If $|E|\geq  \frac{a_{\epsilon_0}+a_{2\epsilon_0}}{2},$
  one has
\begin{eqnarray*}
  \mbox{dist}(E,S_{\epsilon_0}) &= &\inf_{x\in\mathbb{R}}|E  -a_1e^{-2\pi \epsilon_0}e^{2\pi ix} -a_1e^{2\pi \epsilon_0}e^{-2\pi ix} |  \\
  &= & |E  -a_1e^{-2\pi \epsilon_0} -a_1e^{2\pi \epsilon_0} |\\
   &\geq &  \frac{19}{60}a_1 e^{4\pi \epsilon_0} .
\end{eqnarray*}
If $\frac{a_{2\epsilon_0}^2-b_{2\epsilon_0}^2}{a_{2\epsilon_0}}\leq|E|\leq  \frac{a_{\epsilon_0}+a_{2\epsilon_0}}{2}$,   
we have
\begin{eqnarray*}
  \mbox{dist}(E,S_{2\epsilon_0}) &= &\inf_{x\in\mathbb{R}}|E  -a_1e^{-4\pi \epsilon_0}e^{2\pi ix} -a_1e^{4\pi \epsilon_0}e^{-2\pi ix} |  \\
  &= & |E  -a_1e^{-4\pi \epsilon_0} -a_1e^{4\pi \epsilon_0} |\\
   &\geq &   \frac{19}{60}a_1 e^{4\pi \epsilon_0} .
\end{eqnarray*}
If $|E|\leq\frac{a_{2\epsilon_0}^2-b_{2\epsilon_0}^2}{a_{2\epsilon_0}} $,   
we also have
\begin{eqnarray*}
  \mbox{dist}(E,S_{2\epsilon_0}) &= &|E  -a_1e^{-4\pi \epsilon_0}e^{2\pi ix_E} -a_1e^{4\pi \epsilon_0}e^{-2\pi ix_E} |  \\
   &\geq &  \frac{19}{60}a_1 e^{4\pi \epsilon_0}  ,
\end{eqnarray*}
where $\cos 2 \pi x_E=\frac{a_{2\epsilon_0} E}{ a_{2\epsilon_0}^2-b_{2\epsilon_0}^2}$.
\par
Putting all the cases together, we get
\begin{equation*}
   \sup_{\delta\in \{\epsilon_0,2\epsilon_0\}} \inf_{x\in\mathbb{R}}|E  -a_1e^{-2\pi \delta}e^{2\pi ix} -a_1e^{2\pi \delta}e^{-2\pi ix} |\geq \frac{19}{60}a_1 e^{4\pi \epsilon_0} .
 \end{equation*}
\end{proof}

\textbf{Proof of Theorem \ref{Maintheorem}}
\begin{proof}

Since the Lyapunov exponent $L(\alpha,A)$ is upper semi-continuous with respect to $\alpha$ \cite{craig1983subharmonicity},  we only need to prove Theorem \ref{Maintheorem} for
$\alpha\in \mathbb{R}\backslash \mathbb{Q}$.
Without loss of generality, assume $a_1,a_2>0$.
Let  $\epsilon_0>0$ be such that $e^{4\pi \epsilon_0}=(\frac{a_1}{a_2})^{\frac{1}{2}}\geq 10$.
 Applying  Lemma \ref{LeMinestimate}, we have  for any $E\in \mathbb{R}$, there exists $\delta\in \{\epsilon_0,2\epsilon_0\}$ such that
 \begin{equation}\label{2case}
     \inf_{x\in\mathbb{R}}|E  -a_1e^{-2\pi \delta}e^{2\pi ix} -a_1e^{2\pi \delta}e^{-2\pi ix} |\geq \frac{19}{60}a_1 e^{4\pi \epsilon_0}.
 \end{equation}

 Case 1 : $\delta=\epsilon_0$.

   We first   estimate  the left upper element of matrix $A$.
Computing directly, one has for any $E, x\in \mathbb{R}$,
 \begin{equation*}
 |E- 2a_1\cos2\pi (x+i\epsilon_0)-2a_2\cos4\pi (x+i\epsilon_0)|\;\;\;\;\;\;\;\;\;\;\;\;\;\;\;\;\;\;\;\;\;\;\;\;\;\;\;\;\;\;\;\;\;\;\;\;\;\;\;\;\;\;\;\;\;\;\;\;\;\;\;\;\;\;\;\;\;\;\;\;\;\;\;\;\;\;
 \end{equation*}
 \begin{eqnarray}
 \nonumber
     &=& |E-  a_1e^{2\pi \epsilon_0}e^{-2\pi i x}- a_1e^{-2\pi \epsilon_0}e^{2\pi ix}-  a_2e^{4\pi \epsilon_0}e^{-4\pi ix}- a_2e^{-4\pi \epsilon_0}e^{4\pi ix}| \\
     \nonumber
     &\geq&  \frac{19}{ 60 }a_1e^{4\pi\epsilon_0}-a_2e^{4\pi \epsilon_0}-a_2e^{-4\pi \epsilon_0}\\
     \nonumber
      &\geq& a_1e^{4\pi \epsilon_0}(\frac{19}{60 }- \frac{a_2}{a_1}-\frac{a_2}{a_1}e^{-8\pi \epsilon_0})  \\
       \nonumber
     &>&    2 . 
 \end{eqnarray}
 By Lemma \ref{EstimateA_j}, we have   that the left upper element of $A_n(x+i\epsilon_0)$ satisfies the following estimate
  \begin{equation*}
     |(A_n(x+i\epsilon_0)\left(
                           \begin{array}{c}
                             1 \\
                             0 \\
                           \end{array}
                         \right),\left(
                                   \begin{array}{c}
                                     1 \\
                                     0 \\
                                   \end{array}
                                 \right))|\;\;\;\;\;\;\;\;\;\;\;\; \;\;\;\;\;\;\;\;\;\;\;\;\;\;\;\;\;\;\;\;\;\;\;\;\;\;\;\;\;\;\;\;\;\;\;\;\;\;\;\;\;\;\;\;\;\;\;\;\;
  \end{equation*}
  \begin{equation*}
    \;\;\;\;\;\;\;\;\;\;\;\;\;\;   \geq
     \prod_{j=0}^{n-1} (|E- 2a_1\cos2\pi (x+i\epsilon_0+j\alpha)-2a_2\cos4\pi (x+i\epsilon_0+j\alpha)|-1).
  \end{equation*}
   This implies
  \begin{equation}\label{Estepsilon_0}
    ||A_n(x+i\epsilon_0) || \geq
     \prod_{j=0}^{n-1} (|E- 2a_1\cos2\pi (x+i\epsilon_0+j\alpha)-2a_2\cos4\pi (x+i\epsilon_0+j\alpha)|-1).
  \end{equation}
  Furthermore, we have
  \begin{eqnarray*}
    L^E(\epsilon_0) &=& \lim_{n\rightarrow \infty}\frac{1}{n}\int_{\mathbb{T}} \ln  ||A_n(x+i\epsilon_0) || dx \\
     &\geq&  \int_{\mathbb{T}} \ln  (|E- 2a_1\cos2\pi (x+i\epsilon_0)-2a_2\cos4\pi (x+i\epsilon_0)|-1)dx\\
      &\geq& \int _{\mathbb{T}}\ln   |E- 2a_1\cos2\pi (x+i\epsilon_0)|dx +\int _{\mathbb{T}}\ln (1- \frac{| 2a_2\cos4\pi (x+i\epsilon_0)|+1  }{|E- 2a_1\cos2\pi (x+i\epsilon_0)|})dx  \\
      &=&  I+II.
  \end{eqnarray*}
  Now we will estimate I and II separately.
  First, one has
  \begin{eqnarray*}
    II &\geq& \ln (1-\frac{60}{19} \frac{(a_2e^{4\pi \epsilon_0}+a_2e^{-4\pi \epsilon_0}+1)}{a_1e^{4\pi \epsilon_0}})\\
     &\geq& - 10 (\frac{a_2}{a_1}  )^{\frac{1}{2}},
  \end{eqnarray*}
where the first inequality holds  by (\ref{2case}).
\par
We use Herman's subharmonic  method to estimate I:
\begin{eqnarray*}
  I &=&  \int _{\mathbb{T}}\ln   |E-  a_1e^{2\pi \epsilon_0}e^{-2\pi ix}-  a_1e^{-2\pi \epsilon_0}e^{2\pi ix}|dx  \\
    &=&   \int _{|z|=1}\ln   |E z-  a_1e^{2\pi \epsilon_0}- a_1e^{-2\pi \epsilon_0}z^2|dz \\
   &\geq & 2\pi \epsilon_0+   \ln a_1 .
\end{eqnarray*}
Thus
\begin{equation*}
   L^E(\epsilon_0)\geq  2\pi \epsilon_0+   \ln a_1 - 10 (\frac{a_2}{a_1}  )^{\frac{1}{2}}.
\end{equation*}
Suppose point $(\epsilon_0, L^E(\epsilon_0))\in \ell_2$. Then
one has
\begin{equation*}
    \ln a_2 +4\pi \epsilon_0\geq 2\pi \epsilon_0+   \ln a_1- 10 (\frac{a_2}{a_1}  )^{\frac{1}{2}}.
\end{equation*}
This implies that
\begin{equation*}
      \frac{40}{3} (\frac{a_2}{a_1})^{\frac{1}{2}}\geq  \ln  \frac{a_1}{a_2}.
\end{equation*}
However this is impossible because $a_2\leq \frac{a_1}{100}$.
\par
Now that  $(\epsilon_0, L^E(\epsilon_0))\notin \ell_2$,
we must have
\begin{eqnarray*}
  L^E +2\pi \epsilon_0 &\geq&  L^E(\epsilon_0) \\
   &\geq&   2\pi \epsilon_0+   \ln a_1- 10 (\frac{a_2}{a_1}  )^{\frac{1}{2}},
\end{eqnarray*}
This implies
\begin{equation*}
   L^E \geq \ln a_1 - 10 (\frac{a_2}{a_1}  )^{\frac{1}{2}}.
\end{equation*}

\par
Case 2: $\delta=2\epsilon_0$.

In this case, we have
 \begin{equation*}
     \inf_{x\in\mathbb{R}}|E  -a_1e^{-4\pi \epsilon_0}e^{2\pi ix} -a_1e^{4\pi\epsilon_0}e^{-2\pi ix} |\geq \frac{19}{60}a_1 e^{4\pi \epsilon_0} .
 \end{equation*}
 The proof of case 2 is similar to  that of case 1.  We give the details below.
Computing directly, one has for any $E, x\in \mathbb{R}$,
 \begin{equation*}
 |E- 2a_1\cos2\pi (x+2i\epsilon_0)-2a_2\cos4\pi (x+2i\epsilon_0)|\;\;\;\;\;\;\;\;\;\;\;\;\;\;\;\;\;\;\;\;\;\;\;\;\;\;\;\;\;\;\;\;\;\;\;\;\;\;\;\;\;\;\;\;\;\;\;\;\;\;\;\;\;\;\;\;\;\;\;\;\;\;\;\;\;\;
 \end{equation*}
 \begin{eqnarray*}
     &=& |E-  a_1e^{4\pi \epsilon_0}e^{-2\pi i x}- a_1e^{-4\pi \epsilon_0}e^{2\pi ix}-  a_2e^{\pi \epsilon_0}e^{-4\pi ix}- a_2e^{-8\pi \epsilon_0}e^{4\pi ix}| \\
     &\geq&  \frac{19}{ 60 }a_1e^{4\pi\epsilon_0}-a_2e^{8\pi \epsilon_0}-a_2e^{-8\pi \epsilon_0}\\
      &\geq& a_1e^{4\pi \epsilon_0}(\frac{19}{ 60  }- e^{4\pi \epsilon_0}\frac{a_2}{a_1}-e^{-12\pi \epsilon_0}\frac{a_2}{a_1})  \\
     &>&    2 .
 \end{eqnarray*}
As in (\ref{Estepsilon_0}), we have
  \begin{equation*}
    ||A_n(x+2i\epsilon_0) || \geq
     \prod_{j=0}^{n-1} (|E- 2a_1\cos2\pi (x+2i\epsilon_0+j\alpha)-2a_2\cos4\pi (x+2i\epsilon_0+j\alpha)|-1).
  \end{equation*}
  Then we have
  \begin{eqnarray*}
    L^E(2\epsilon_0) &=& \lim_{n\rightarrow \infty}\frac{1}{n}\int _{\mathbb{T}}\ln  ||A_n(x+2i\epsilon_0) || dx \\
     &\geq&  \int _{\mathbb{T}}\ln  (|E- 2a_1\cos2\pi (x+2i\epsilon_0)-2a_2\cos4\pi (x+2i\epsilon_0)|-1)dx\\
      &\geq& \int_{\mathbb{T}} \ln   |E- 2a_1\cos2\pi (x+2i\epsilon_0)|dx +\int _{\mathbb{T}}\ln (1- \frac{| 2a_2\cos4\pi (x+2i\epsilon_0)|+1  }{|E- 2a_1\cos2\pi (x+2i\epsilon_0)|})dx  \\
      &=&  I+II.
  \end{eqnarray*}
 Following the discussion of  case 1, I and II have the following lower bounds:
  \begin{eqnarray*}
    II &\geq& \ln (1- \frac{60}{19}\frac{(a_2e^{8\pi \epsilon_0}+a_2e^{-8\pi \epsilon_0}+1)}{a_1e^{4\pi \epsilon_0}})\\
     &\geq& - 10 (\frac{a_2}{a_1}) ^{\frac{1}{2}} .
  \end{eqnarray*}
and
\begin{eqnarray*}
  I &=&  \int _{\mathbb{T}} \ln   |E-  a_1e^{4\pi \epsilon_0}e^{-2\pi ix}-  a_1e^{-4\pi \epsilon_0}e^{2\pi ix}|dx  \\
   &\geq & 4\pi \epsilon_0+   \ln a_1 .
\end{eqnarray*}
Thus
\begin{equation*}
   L^E(2\epsilon_0)\geq  4\pi \epsilon_0+   \ln a_1-   10 (\frac{a_2}{a_1})^{\frac{1}{2}}.
\end{equation*}
Suppose point $(2\epsilon_0, L^E(2\epsilon_0))\in \ell_2$. Then
one has
\begin{equation*}
    \ln a_2 +8\pi \epsilon_0\geq 4\pi \epsilon_0+   \ln a_1-  10 (\frac{a_2}{a_1})^{\frac{1}{2}}.
\end{equation*}
This implies that
\begin{equation*}
     20 ( \frac{a_2}{a_1})^{\frac{1}{2}}\geq  \ln  \frac{a_1}{a_2},
\end{equation*}
impossible because $a_2\leq \frac{a_1}{100}$.
\par
Now that   $(2\epsilon_0, L^E(2\epsilon_0))\notin \ell_2$,
we must have
\begin{eqnarray*}
  L^E +4\pi \epsilon_0 &\geq&  L^E(2\epsilon_0) \\
   &\geq&   4\pi \epsilon_0+   \ln a_1-   10(\frac{a_2}{a_1})^{\frac{1}{2}}.
\end{eqnarray*}
therefore
\begin{equation*}
   L^E \geq \ln a_1-   10(\frac{a_2}{a_1})^{\frac{1}{2}}.
\end{equation*}

\end{proof}
 \section*{Acknowledgments}
S.J. is a 2014-15 Simons Fellow. This research was partially
 supported by NSF DMS-1401204.

\end{document}